\documentclass[10pt]{article}
\usepackage{amsmath}
\usepackage{amssymb, amscd, amsthm}
\usepackage[all]{xy}
\usepackage[dvips]{graphicx}
\usepackage{verbatim}
\usepackage[perpage,symbol*]{footmisc}
\usepackage{cite}
\newtheorem{lemma}{\textbf{Lemma}}[section]
\newtheorem{theorem}{\textbf{Theorem}}
\newtheorem{remark}{\textbf{Remark}}[section]
\newtheorem{corollary}{\textbf{Corollary}}[section]

\newtheorem{example}{\textbf{Example}}[section]

\usepackage{longtable}
\newcommand{\tabincell}[2]{\begin{tabular}{@{}#1@{}}#2\end{tabular}}
\usepackage{multirow}

\newcommand{\F}{\mathbb{F}}

\begin{document}

\baselineskip 17pt\title{\bf New Galois Hulls of GRS Codes and Application to EAQECCs}
\author{\large  Xiaolei Fang\quad\quad Renjie Jin \quad\quad Jinquan Luo* \quad\quad Wen Ma}\footnotetext{Xiaolei Fang is with School of Mathematics and Computer Science, Wuhan Polytechnic University, Wuhan China 430048. Renjie Jin is with College of Liberal Arts and Sciences, National University of Defense Technology, Changsha China 410073. Jinquan Luo and Wen Ma are with School of Mathematics and Statistics \& Hubei Key Laboratory of Mathematical Sciences, Central China Normal University, Wuhan China 430079.\\  E-mail: fxl\_1918@163.com(X.Fang), jinrenjie@mails.ccnu.edu.cn(R.Jin), luojinquan@mail.ccnu.edu.cn(J.Luo), mawen95@126.com(W.Ma) \\ *: Corresponding author}
\date{}
\maketitle

{\bf Abstract:} Galois hulls of linear codes have important applications in quantum coding theory. In this paper, we construct some new classes of (extended) generalized Reed-Solomon (GRS) codes with Galois hulls of arbitrary dimensions. We also propose a general method on constructing GRS codes with Galois hulls of arbitrary dimensions from special Euclidean orthogonal GRS codes. Finally, we construct several new families of entanglement-assisted quantum error-correcting codes (EAQECCs) and MDS EAQECCs by utilizing the above results.



\section{Introduction}
 \quad\; Let $q=p^h$, where $p$ is an odd prime. Denote by $\mathbb{F}_{q}$ the finite field with $q$ elements. An $[n,k,d]_{q}$ code $\mathcal{C}$ is
a linear code over $\mathbb{F}_{q}$ with length $n$, dimension $k$ and minimum distance $d$. The Singleton bound states that $k\leq n-d+1$.
If the parameters of $\mathcal{C}$ reach the Singleton bound (i.e., $k=n-d+1$), then $\mathcal{C}$ is called a maximum distance separable (MDS) code.
Due to the optimal properties, MDS codes play an important role in coding theory and related fields, see [\ref{BR}, \ref{SR}].

Let $\mathcal{C}^\bot$ be the dual code of $\mathcal{C}$. The hull of $\mathcal{C}$ is defined by $Hull(\mathcal{C})=\mathcal{C}\bigcap \mathcal{C}^\bot$.
Two special cases are extremely interesting for researchers. One is $Hull(\mathcal{C})=\{ \mathbf{0}\}$. In this case, $\mathcal{C}$ is called a linear
complementary dual (LCD) code. In [\ref{CG}], Carlet et al. constructed LCD codes by utilizing cyclic
codes, Reed-Solomon codes and generalized residue codes, together with direct sum, puncturing, shortening, extension, $(u|u+v)$ construction and suitable
automorphism action. In [\ref{CMTQ}] and [\ref{CMTQP}], Carlet et al. showed that any linear code over $\mathbb{F}_{q}$ $(q>3)$ is equivalent to a
Euclidean LCD code and any linear code over $\mathbb{F}_{q^{2}}$ $(q>2)$ is equivalent to a Hermitian LCD code. The other case is
$Hull(\mathcal{C})=\mathcal{C}$ (resp. $\mathcal{C}^\bot$), in which $\mathcal{C}$ is called a self-orthogonal (resp. dual containing) code.
In particular, the code $\mathcal{C}$ satisfying $\mathcal{C}=\mathcal{C}^\bot$ is called a self-dual code. The construction of MDS Euclidean self-dual code is a popular issue in recent years.
In [\ref{GG}], Grassl and Gulliver showed that the problem has been completely solved over finite fields of characteristic $2$. In [\ref{JX}]
and [\ref{Yan}], Jin, Xing and Yan constructed some classes of new MDS self-dual codes through (extended) GRS codes.

In 1995, the first quantum error-correcting code was constructed. Shortly after that, Shor et al. established a connection between quantum codes and
classical error-correcting codes satisfying certain self-orthogonal or dual containing property in [\ref{SL}]. In [\ref{BDH}], Brun et al. introduced
entanglement-assisted quantum error-correcting code (EAQECC), which did not require the dual-containing property for standard quantum error-correcting
codes. Thus we could construct EAQECCs via classical linear codes without self-orthogonality. However, the determination of the number of shared pairs
was not an easy thing. In [\ref{GJG}], Guenda et al. showed that the number of shared pairs was related to the hull of classical linear code. So
the results on hulls of classical linear codes can be applied in the constructions of EAQECCs and some new families of EAQECCs were discovered through
the hulls of classical linear codes.

Some people have studied the hulls of classical linear codes and constructed EAQECCs and MDS EAQECCs. In [\ref{LC}] and [\ref{LCC}], several
infinite families of MDS codes with Euclidean hulls of arbitrary dimensions were presented. Then they were applied to construct some families of
MDS EAQECCs. Also utilizing (extended) GRS codes, in [\ref{FFLZ}], several new MDS codes with Euclidean or Hermitian hulls of arbitrary dimensions
were proposed and these MDS codes were also applied in the constructions of MDS EAQECCs. Recently, in [\ref{Cao}],
Cao gave several new families of MDS codes with Galois hulls of arbitrary dimensions and constructed nine new families of MDS EAQECCs.
In [\ref{QCWL}], Qian et al. gave a general construction on MDS codes with Galois hulls of arbitrary dimensions.

Based on the above results, we construct some new MDS codes with Galois hulls of arbitrary dimensions. Furthermore, we propose a mechanism on how to find
MDS codes with Galois hulls of arbitrary dimensions from GRS codes with special Euclidean orthogonal property. All the known results on the constructions
of Galois hulls of MDS codes are listed in Table 1.

This paper is organized as follows. In Section 2, we will introduce some basic knowledge and useful results on GRS codes and Galois hulls.
In Section 3, we construct some new (extended) GRS codes with Galois hulls of arbitrary dimensions. In Section 4, we propose a mechanism for
the constructions of some special cases. In Section 5, we will present our main results on the constructions of new EAQECCs and MDS EAQECCs.
In Section 6, we will make a conclusion.

\begin{center}
\begin{longtable}{|c|c|c|c|}  
\caption{Known constructions on Galois hulls of MDS codes} \\ \hline
$q$ & $n$ & $k$ & Reference\\  \hline
$q=p^h$  &  $n \mid q-1$  & $1\leq k\leq \lfloor\frac{p^e+n-1}{p^e+1}\rfloor$  & [\ref{Cao}] \\ \hline
$q=p^h$ & $n|(p^e-1)$ &  $1\leq k\leq \lfloor\frac{n}{2}\rfloor$  & [\ref{Cao}]\\ \hline

$q=p^h$ is odd, $2e\mid h$ &  $n\leq p^e$ & $1\leq k\leq \lfloor\frac{p^e+n-1}{p^e+1}\rfloor$  & [\ref{Cao}]\\ \hline

$q=p^h$ is odd, $2e\mid h$ & \tabincell{c}{$n=\frac{r(q-1)}{\gcd(x_2,q-1)}$, $1\leq r\leq \frac{q-1}{\gcd(x_1,q-1)}$, \\ $(q-1)|lcm(x_1,x_2)$ and $\frac{q-1}{p^e-1}\mid x_1$} & $1\leq k\leq \lfloor\frac{p^e+n}{p^e+1}\rfloor$  &  [\ref{Cao}] \\ \hline

$q=p^h$ is odd, $2e\mid h$ & \tabincell{c}{$n=\frac{r(q-1)}{\gcd(x_2,q-1)}+1$, $1\leq r\leq \frac{q-1}{\gcd(x_1,q-1)}$, \\ $(q-1)|lcm(x_1,x_2)$ and $\frac{q-1}{p^e-1}\mid x_1$} & $1\leq k\leq \lfloor\frac{p^e+n}{p^e+1}\rfloor$  &  [\ref{Cao}] \\ \hline

$q=p^h$ is odd, $2e\mid h$ & \tabincell{c}{$n=\frac{r(q-1)}{\gcd(x_2,q-1)}+2$, $1\leq r\leq \frac{q-1}{\gcd(x_1,q-1)}$, \\ $(q-1)|lcm(x_1,x_2)$ and $\frac{q-1}{p^e-1}\mid x_1$} & $1\leq k\leq \lfloor\frac{p^e+n}{p^e+1}\rfloor$  &  [\ref{Cao}] \\ \hline

$q=p^h$ is odd, $2e\mid h$ & \tabincell{c}{$n=rm$, $1\leq r\leq \frac{p^e-1}{m_1}$, $m_1=\frac{m}{\gcd(m,y)}$, \\ $m\mid (q-1)$ and $y=\frac{q-1}{p^e-1}$} & $1\leq k\leq \lfloor\frac{p^e+n}{p^e+1}\rfloor$  &  [\ref{Cao}] \\ \hline

$q=p^h$ is odd $2e\mid h$ & \tabincell{c}{$n=rm+1$, $1\leq r\leq \frac{p^e-1}{m_1}$, $m_1=\frac{m}{\gcd(m,y)}$, \\ $m\mid (q-1)$ and $y=\frac{q-1}{p^e-1}$} & $1\leq k\leq \lfloor\frac{p^e+n}{p^e+1}\rfloor$  &  [\ref{Cao}] \\ \hline

$q=p^h$ is odd, $2e\mid h$ & \tabincell{c}{$n=rm+2$, $1\leq r\leq \frac{p^e-1}{m_1}$, $m_1=\frac{m}{\gcd(m,y)}$, \\ $m\mid (q-1)$ and $y=\frac{q-1}{p^e-1}$} & $1\leq k\leq \lfloor\frac{p^e+n}{p^e+1}\rfloor$  &  [\ref{Cao}] \\ \hline

$q = p^h $ is even & $n\leq q$, $\frac{m}{\gcd(e,m)}$ and $m>1$ &  $1\leq k\leq \lfloor\frac{p^e+n-1}{p^e+1}\rfloor$  & [\ref{QCWL}] \\ \hline

$q=p^h>3$ & $n\leq r$, $r=p^m$ with $m\mid h$ and $p^e+1\mid\frac{q-1}{r-1}$ &  $1\leq k\leq \lfloor\frac{p^e+n-1}{p^e+1}\rfloor$ & [\ref{QCWL}]\\\hline

$q = p^h>3 $ & $n\mid q$ & $1\leq k\leq \lfloor\frac{p^e+n-1}{p^e+1}\rfloor$   & [\ref{QCWL}] \\ \hline

$q=p^h>3$ & \tabincell{c}{$p\mid n$, $F(x)=a+bx+\sum_{i:p|i}a_ix^i+x^n$ \\ can be completely decomposed in $\F_q$} & $1\leq k\leq \lfloor\frac{p^e+n-1}{p^e+1}\rfloor$  & [\ref{QCWL}] \\ \hline

$q=p^h>3$ & $(n-1)|(q -1)$  &  $1\leq k\leq \lfloor\frac{p^e+n-1}{p^e+1}\rfloor$   &  [\ref{QCWL}] \\ \hline

$q=p^h>3$ & $n=2n'$ and $n'|q$   &  $1\leq k\leq \lfloor\frac{p^e+n-1}{p^e+1}\rfloor$   & [\ref{QCWL}]\\\hline

$q=p^h>3$ & \tabincell{c}{$n= n't$, $n'=r^m$, $r=p^\varepsilon$, $1\leq t \leq r$, \\ $1\leq m\leq \frac{h}{\varepsilon}-1$ and $\gcd(p^e+1,q-1)\mid\frac{q-1}{r-1}$} &  $1\leq k\leq \lfloor\frac{p^e+n-1}{p^e+1}\rfloor$   & [\ref{QCWL}]\\ \hline

$q=p^{h}>3$ & $n\mid (q-1)$ & $1\leq k\leq \lfloor\frac{p^e+n}{p^e+1}\rfloor$   & [\ref{QCWL}]\\ \hline

$q=p^h$ & $n=2n'$ and $n'|(q-1)$   &  $1\leq k\leq \lfloor\frac{p^e+n}{p^e+1}\rfloor$   & [\ref{QCWL}]\\\hline

$q=p^h>3$ & \tabincell{c}{$n=tn'$, $n'|(q-1)$, $1\leq t\leq \frac{r-1}{n_1}$, \\ $r=p^m$ with $m\mid h$, $n_1=\frac{n'}{n_2}$,  \\ $n_2=\gcd(n',\frac{q-1}{r-1})$ and  $\gcd(p^e+1,q-1)\mid\frac{q-1}{r-1}$} &  $1\leq k\leq \lfloor\frac{p^e+n}{p^e+1}\rfloor$   & [\ref{QCWL}]\\\hline
\end{longtable}
\end{center}

\section{Preliminaries}

\quad\; In this section, we introduce some basic notations and useful results on (extended) GRS codes and Galois hulls.
Readers are referred to [\ref{MS}, Chapter 10] for more details on (extended) GRS codes.

Let $\mathbb{F}_{q}$ be a finite field with $q$ elements. Denote by $\F_q^*=\F_q\backslash\{0\}$. In this paper, we always assume $q=p^h$,
where $p$ is an odd prime. For $1\leq n\leq q$, choose two vectors
$\mathbf{v}=(v_{1},v_{2},\ldots,v_{n})\in (\mathbb{F}_{q}^{*})^{n}$ and $\mathbf{a}=(a_{1},a_{2},\ldots,a_{n})\in \mathbb{F}_{q}^{n}$,
where $a_{i}(1\leq i\leq n)$ are distinct. For an integer $k$ with $1\leq k\leq n$, the GRS code of length $n$ associated with $\mathbf{v}$
and $\mathbf{a}$ is defined as follows:
\begin{equation}\label{def GRS}
\mathbf{GRS}_{k}(\mathbf{a},\mathbf{v})=\{(v_{1}f(a_{1}),\ldots,v_{n}f(a_{n})):f(x)\in\mathbb{F}_{q}[x], \mathrm{deg}(f(x))\leq k-1\}.
\end{equation}
The code $\mathbf{GRS}_{k}(\mathbf{a},\mathbf{v})$ is a $q$-ary $[n,k]$ MDS code and its dual is also MDS [\ref{MS}, Chapter 11].

The extended GRS code associated with $\mathbf{v}$ and $\mathbf{a}$ is defined by:
\begin{equation}\label{def extended GRS}
\mathbf{GRS}_{k}(\mathbf{a},\mathbf{v},\infty)=\{(v_{1}f(a_{1}),\ldots,v_{n}f(a_{n}),f_{k-1}):f(x)\in\mathbb{F}_{q}[x],
\mathrm{deg}(f(x))\leq k-1\},
\end{equation}
where $f_{k-1}$ is the coefficient of $x^{k-1}$ in $f(x)$.
The code $\mathbf{GRS}_{k}(\mathbf{a},\mathbf{v},\infty)$ is a $q$-ary $[n+1,k]$ MDS code and its dual is also MDS [\ref{MS}, Chapter 11].

For $1\leq i\leq n$, we define
\begin{equation*}\label{solution}
u_i:=\prod_{1\leq j\leq n,j\neq i}(a_{i}-a_{j})^{-1},
\end{equation*}
which will be used frequently in this paper.

For any two vectors $\mathbf{a}=(a_{1},a_{2},\ldots,a_{n})$ and $\mathbf{b}=(b_{1},b_{2},\ldots,b_{n})$ in $\mathbb{F}_{q}^{n}$,
we define their Euclidean inner product as:
\begin{equation*}
\langle\mathbf{a},\mathbf{b}\rangle=\sum\limits_{i=1}^{n}a_{i}b_{i}.
\end{equation*}
The Euclidean dual code of $\mathcal{C}$ is defined as:
\begin{equation*}
\mathcal{C}^\perp=\left\{\mathbf{a} \in \mathbb{F}_{q}^{n}:\langle\mathbf{a},\mathbf{b}\rangle=0  \text{ for any }\mathbf{b} \in \mathcal{C} \right\}.
\end{equation*}
The Euclidean hull of $\mathcal{C}$ is defined by $Hull(\mathcal{C})=\mathcal{C}\cap \mathcal{C}^\perp$.

Similarly, if $h$ is even, for any two vectors $\mathbf{a}=(a_{1},a_{2},\ldots,a_{n})$ and $\mathbf{b}=(b_{1},b_{2},\ldots,b_{n})$ in
$\mathbb{F}_{q}^{n}$, the Hermitian inner product is defined as:
\begin{equation*}
\langle\mathbf{a},\mathbf{b}\rangle_{\frac{h}{2}}=\sum\limits_{i=1}^{n}a_{i}b_{i}^{p^{\frac{h}{2}}}.
\end{equation*}
The Hermitian dual code of $\mathcal{C}$ is defined as:
\begin{equation*}
\mathcal{C}^{\perp_{\frac{h}{2}}}=\left\{\mathbf{a} \in \mathbb{F}_{q}^{n}:\langle\mathbf{a},\mathbf{b}\rangle_{\frac{h}{2}}=0  \text{ for any }\mathbf{b} \in \mathcal{C} \right\}.
\end{equation*}
The Hermitian hull of $\mathcal{C}$ is defined by $Hull_{\frac{h}{2}}(\mathcal{C})=\;\mathcal{C}\cap \mathcal{C}^{\perp_{\frac{h}{2}}}.$

In [\ref{FZ}], Fan et al. gave the definition of $e$-Galois inner product with $0\leq e\leq h-1$. It is a generalization of Euclidean inner product and
Hermitian inner product. For any two vectors $\mathbf{a}=(a_{1},a_{2},\ldots,a_{n})$ and $\mathbf{b}=(b_{1},b_{2},\ldots,b_{n})$ in
$\mathbb{F}_{q}^{n}$, the $e$-Galois inner product is defined as:
\begin{equation*}
\langle\mathbf{a},\mathbf{b}\rangle_e=\sum\limits_{i=1}^{n}a_{i}b_{i}^{p^e}.
\end{equation*}
The $e$-Galois dual code of $\mathcal{C}$ is defined as:
\begin{equation*}
\mathcal{C}^{\perp_e}=\left\{\mathbf{a} \in \mathbb{F}_{q}^{n}:\langle\mathbf{a},\mathbf{b}\rangle_e=0  \text{ for any }\mathbf{b} \in \mathcal{C} \right\}.
\end{equation*}
The $e$-Galois hull of $\mathcal{C}$ is defined by $Hull_{e}(\mathcal{C})=\;\mathcal{C}\cap \mathcal{C}^{\perp_{e}}.$

In the following, we list some useful results, which will be used in our constructions.

\begin{lemma}([\ref{Cao}, Lemmas 2 and 3])\label{HE1}
Let $\mathcal{C}$ be a linear code and $\mathbf{c}$ be a codeword of $\mathcal{C}$.

(i). For $\mathcal{C}=\mathbf{GRS}_{k}(\mathbf{a},\mathbf{v})$, the codeword
$\mathbf{c}=(v_{1}f(a_{1}),\ldots,v_{n}f(a_{n}))\in\mathcal{C}\bigcap\mathcal{C}^{\bot_e}$ if and only if there exists a polynomial
$g(x)\in\mathbb{F}_{q}[x]$ with $\mathrm{deg}(g(x))\leq n-k-1$, such that
\begin{center}
$(v_{1}^{p^e+1}f^{p^e}(a_{1}),v_{2}^{p^e+1}f^{p^e}(a_{2}),\ldots,v_{n}^{p^e+1}f^{p^e}(a_{n}))=(u_{1}g(a_{1}),u_{2}g(a_{2}),\ldots,u_{n}g(a_{n}))$.
\end{center}

(ii). For $\mathcal{C}=\mathbf{GRS}_{k}(\mathbf{a}, \mathbf{v}, \infty)$, the codeword
$\mathbf{c}=(v_{1}f(a_{1}),\ldots,v_{n}f(a_{n}), f_{k-1})\in\mathcal{C}\bigcap\mathcal{C}^{\bot_e}$
if and only if there exists a polynomial $g(x)\in\mathbb{F}_{q}[x]$ with $\mathrm{deg}(g(x))\leq n-k$, such that
\begin{center}
$(v_{1}^{p^e+1}f^{p^e}(a_{1}),v_{2}^{p^e+1}f^{p^e}(a_{2}),\ldots,v_{n}^{p^e+1}f^{p^e}(a_{n}),f_{k-1}^{p^e})=
(u_{1}g(a_{1}),u_{2}g(a_{2}),\ldots,u_{n}g(a_{n}), -g_{n-k})$.
\end{center}
\end{lemma}

\begin{lemma}([\ref{FMJ}, Corollaries 2.1 and 2.2])\label{cor1}
Let $\mathcal{C}(m)$ be an (extended) GRS code, where $m$ represents dimension.

(i). Assume that $\mathcal{C}(m)=\mathbf{GRS}_{m}(\mathbf{a},\mathbf{v})$ and $1\leq m\leq \lfloor\frac{n}{2}\rfloor$. Then
$\mathcal{C}(m)^\bot=\mathcal{C}(n-m)$ if and only if there exists $\lambda\in \mathbb{F}_{q}^{*}$ such that $\lambda u_{i}=v_{i}^{2}$
for all $i=1,2,\ldots,n$. In particular, when $n$ is even and $m=\frac{n}{2}$, $\mathcal{C}(\frac{n}{2})$ is MDS self-dual (see [\ref{JX}, Corollary 2.4]).

(ii). Assume that $\mathcal{C}(m)=\mathbf{GRS}_{m}(\mathbf{a},\mathbf{v},\infty)$ and $1\leq m\leq \lfloor\frac{n+1}{2}\rfloor$.
Then $\mathcal{C}(m)^{\bot}=\mathcal{C}(n+1-m)$ if and only if $-u_{i}=v_{i}^{2}$ for all $i=1,2 ,\ldots,n$. In particular, when $n$ is odd and
$m=\frac{n+1}{2}$, $\mathcal{C}(\frac{n+1}{2})$ is MDS self-dual (see Lemma 2.2 of [\ref{Yan}]).
\end{lemma}

Denote by $E:=\{x^{p^e+1}|x\in\F_q^*\}$. In fact, $E$ is a multiplicative subgroup of $\F_q^*$ (see [\ref{Cao}]). We have the following result.

\begin{lemma}([\ref{Cao}, Lemma 4])\label{yC}
Let $q=p^h$ and $0\leq e\leq h-1$. Then $\F_{p^e}^*\subseteq E$ if and only if $2e\mid h$.
\end{lemma}

At the end of this section, we define two polynomials $\Psi_B(x)$ and $\Delta_B(x)$ over $\F_q$ as
\begin{center}
$\Psi_B(x)=\prod\limits_{b\in B}(x-b)$ and $\Delta_B(x)=\Psi_B'(x)$
\end{center}
for any $B\subseteq\F_q$. It is easy to see that $\Delta_B(b)=\Psi_B'(b)=\prod\limits_{b'\in B,b'\neq b}(b-b')$, for any $b\in B$.

\section{Some constructions of MDS codes with $e$-Galois hulls of arbitrary dimensions}

 \quad\; In this section, we will present several new MDS codes with $e$-Galois hulls of arbitrary dimensions via (extended) GRS codes.

 \begin{theorem}\label{thmE}
Let $q=p^{em}$ with $p$ odd prime. Assume that $t\mid(p^{e}-1)$, $m$ is even, $r\leq m-1$ and $n=tp^{er}$. Then for any
$1\leq k\leq \lfloor\frac{p^e+n-1}{p^e+1}\rfloor$ and $0\leq l\leq k-1$, there exists an $[n+1,k]_q$ MDS code $\mathcal{C}$ with $l$-dimensional $e$-Galois hull.
 \end{theorem}
 \begin{proof}
Let $V$ be an $r$-dimensional $\mathbb{F}_{p^e}$-vector subspace in $\mathbb{F}_{q}$ with $V\cap\mathbb{F}_{p^e}=0$. Choose $\omega\in\mathbb{F}_{p^e}^*$,
which is a $t$-th primitive root of unity. Let $V_j=\omega^{j}+V(j=0,1,\ldots,t-1)$ and $\bigcup\limits_{j=0}^{t-1}V_{j}=\{a_1,\ldots,a_n\}$. For any $a_i\in V_{j_0}$,
\begin{equation*}
\begin{aligned}
u_i&=\Delta_{V_{j_0}}(a_i)\prod_{j=0,j\neq j_0}^{t-1}\Psi_{V_j}(a_i)&\\
&=\left(\prod_{0\neq \gamma\in V}\gamma\right)\cdot\left(\prod_{j=0,j\neq j_0}^{t-1}\prod_{\gamma\in V}(\omega^{j_0}-\omega^{j}+\gamma)\right)&\\
&=\left(\prod_{0\neq \gamma\in V}\gamma\right)\cdot\left(\prod_{\gamma\in V}\omega^{j_0(t-1)}\prod_{d=1}^{t-1}\left(1+\omega^{-j_0}\gamma-\omega^{d}\right)\right)&\\
&=\omega^{-j_0p^{er}}\cdot\left(\prod_{0\neq \gamma\in V}\gamma\right)\cdot\left(\prod_{\gamma\in V}\prod_{d=1}^{t-1}(1+\gamma-\omega^{d})\right)&
\end{aligned}
\end{equation*}
where $d=j-j_0$ and the last equality follows from that $\prod\limits_{\gamma\in V}\omega^{j_0(t-1)}=\omega^{-j_0p^{er}}$ and $\omega^{-j_0}\gamma$
runs through $V$ when $\gamma$ runs through $V$.

Let $b=\left(\prod\limits_{0\neq \gamma\in V}\gamma\right)\cdot\left(\prod\limits_{\gamma\in V}\prod\limits_{d=1}^{t-1}(1+\gamma-\omega^{d})\right)$,
which is independent of $i$. It follows that $u_i=\omega^{-j_0p^{er}}\cdot b$, for any $1\leq i\leq n$. Choose $\lambda=b^{-1}\in\F_q^*$. Since
$\omega\in\F_{p^e}^*$ and Lemma \ref{yC}, then $\lambda u_i\in\F_{p^e}^*\subseteq E$ with $1\leq i\leq n$. Let $v_{i}^{p^e+1}=\lambda u_i(1\leq i\leq n)$. Choose
\begin{center}
$\mathbf{a}=(a_1,\ldots,a_{n})$ and
$\mathbf{v}=(\alpha v_{1},\alpha v_{2},\ldots,\alpha v_{s},v_{s+1},\ldots,v_{n})$,
\end{center}
where $\alpha\in\mathbb{F}_{q}^*$ and $\alpha^{p^e+1}\neq 1$. Denote by $\mathcal{C}:=\mathbf{GRS}_k(\mathbf{a},\mathbf{v},\infty)$ and $\xi=\alpha^{p^e+1}$. For any
\begin{center}
$\mathbf{c}=(\alpha v_{1}f(a_{1}),\ldots,\alpha v_{s}f(a_{s}),v_{s+1}f(a_{s+1}), \ldots,v_{n}f(a_{n}),f_{k-1})\in Hull_e(\mathcal{C})$
\end{center}
with $\deg(f(x))\leq k-1$, by Lemma \ref{HE1} (ii), there exists a polynomial $g(x)\in \mathbb{F}_{q}[x]$ with $\deg(g(x))\leq n-k$ such that
\begin{equation}\label{xi}
\begin{aligned}
&(\xi v_{1}^{p^e+1}(a_1)f^{p^e}(a_{1}),\cdots,\xi v_{s}^{p^e+1}(a_s)f^{p^e}(a_{s}),v_{s+1}^{p^e+1}(a_{s+1})f^{p^e}(a_{s+1}),\cdots,
 v_{n}^{p^e+1}(a_{n})f^{p^e}(a_{n}),f^{p^e}_{k-1})&\\
=&(u_{1}g(a_{1}),\cdots,u_{s}g(a_{s}),u_{s+1}g(a_{s+1}),\cdots,u_{n}g(a_{n}),-g_{n-k}).&
\end{aligned}
\end{equation}
From $(\ref{xi})$ and $v_i^{p^e+1}=\lambda u_i$($1\leq i\leq n$), we derive
\begin{equation}\label{xi2}
\begin{aligned}
&(\xi\lambda u_{1}f^{p^e}(a_{1}),\cdots,\xi\lambda u_{s}f^{p^e}(a_s),\lambda u_{s+1}f^{p^e}(a_{s+1}),\cdots, \lambda u_{n}f^{p^e}(a_{n}),f^{p^e}_{k-1})&\\
=&(u_{1}g(a_{1}),\cdots,u_{s}g(a_{s}),u_{s+1}g(a_{s+1}),\cdots,u_{n}g(a_{n}),-g_{n-k}).&
\end{aligned}
\end{equation}
When $s+1\leq i\leq n$, we get $\lambda f^{p^e}(a_i)=g(a_i)$. Note that $\deg(f^{p^e}(x))\leq p^e(k-1)\leq n-k-1$ from
$k\leq \lfloor\frac{p^e+n-1}{p^e+1}\rfloor$ and $\deg(g(x))\leq n-k$. It implies that $\lambda f^{p^e}(x)=g(x)$ from $n-s\geq n-k+1$.

Assume that $f_{k-1}\neq 0$. By $\lambda f^{p^e}(x)=g(x)$, we have $\deg(f^{p^e}(x))=\deg(g(x))$, that is, $p^e(k-1)=n-k$,
which yields a contradiction to $k\leq \lfloor\frac{p^e+n-1}{p^e+1}\rfloor$. Hence, $f_{k-1}=0$, which implies that $\deg(f(x))\leq k-2$.

Comparing the first $s$ coordinates on both sides of (\ref{xi2}), we have $$\xi\lambda u_if^{p^e}(a_i)= u_ig(a_i)= \lambda u_if^{p^e}(a_i)$$ for $i=1,\ldots,s$.
Hence $f^{p^e}(a_i)=0$, i.e., $f(a_i)=0$ for $i=1,\ldots,s$ with $\xi\neq 1$. Then
\begin{equation*}
f(x)=c(x)\prod_{i=1}^{s}(x-a_{i}),
\end{equation*}
for some $c(x)\in \mathbb{F}_{q}[x]$ with $\deg(c(x))\leq k-2-s$. It follows that $\dim(Hull_e(\mathcal{C}))\leq k-1-s$.

Conversely, put $f(x)=c(x)\prod\limits_{i=1}^{s}(x-a_{i})$, where $c(x)\in \mathbb{F}_{q}[x]$ and $\deg(c(x))\leq k-2-s$, which yields $f_{k-1}=0$.
Assume that $g(x)=\lambda f^{p^e}(x)$, then $\deg(g(x))\leq n-k-1$, which yields $g_{n-k}=0$. Therefore,
\begin{equation*}
\begin{aligned}
&(\xi\lambda u_{1}f^{p^e}(a_{1}),\ldots,\xi\lambda u_{s}f^{p^e}(a_{s}),\lambda u_{s+1}f^{p^e}(a_{s+1}),\ldots,\lambda u_{n}f^{p^e}(a_{n}),0)&\\
=&(u_{1}g(a_{1}),\ldots,u_{s}g(a_{s}),u_{s+1}g(a_{s+1}),\ldots,u_{n}g(a_{n}),0).&
\end{aligned}
\end{equation*}
According to Lemma \ref{HE1} (ii),
\begin{center}
$(\alpha v_{1}f(a_{1}),\ldots,\alpha v_{s}f(a_{s}), v_{s+1}f(a_{s+1}),\ldots,  v_{n}f(a_{n}),0)\in Hull_e(\mathcal{C})$.
\end{center}
Thus $\dim(Hull_e(\mathcal{C}))\geq k-1-s$.

Consequently, $\dim(Hull_e(\mathcal{C}))=k-1-s=l$. The proof has been completed.
\end{proof}

Let $1\leq t\leq p^{e}$. We fix an $\F_p$-linear subspace $B\subseteq \F_{p^e}$, such that $|B|\geq t$. Set $b_1=0,b_2,\ldots,b_t$ be $t$ distinct
elements of $B$. Put $Tr$ represents trace function from $\F_q$ to $\F_{p^e}$, with $q=p^h$ and $e\mid h$. Define
\begin{center}
$T_i:=\{x\in\F_q: Tr(x)=b_i\}$,
\end{center}
where $1\leq i\leq t$. Then $|T_i|=p^{h-e}$ and $T_i\bigcap T_j=\emptyset$, for any $1\leq i\neq j\leq t$. The following result can be obtained easily.

\begin{lemma}\label{Psi}
The symbols are the same as above. Then\, $\Delta_{T_i}(x)=1$.
\end{lemma}
\begin{proof}
By the definition of $T_i$ and $Tr$, it is easy to get
\begin{equation}\label{Psiequ}
\Psi_{T_i}(x)=\prod\limits_{\alpha\in T_i}(x-\alpha)=Tr(x)-b_i=x+x^{p^e}+\cdots+x^{p^{h-e}}-b_i.
\end{equation}
Take the derivative of both sides of (\ref{Psiequ}), $$\Delta_{T_i}(x)=\Psi_{T_i}'(x)=1.$$
\end{proof}

\begin{remark}
The case $e=\frac{h}{2}$ has been shown in [\ref{FXF}].
\end{remark}

Now, we give the second construction.
\begin{theorem}
Let $q=p^h$, $2e\mid h$ and $n=tp^{h-e}$ with $1\leq t\leq p^{e}$. For any $1\leq k\leq \lfloor\frac{p^e+n-1}{p^e+1}\rfloor$ and $0\leq l\leq k-1$,
there exists an $[n+1,k]_q$ MDS code $\mathcal{C}$ with $l$-dimensional $e$-Galois hull.
\end{theorem}
\begin{proof}
Since $T_i$ and $T_j$ are pairwise disjoint for any $1\leq i\neq j\leq t$, it follows that $|\bigcup\limits_{i=1}^t T_i|=tp^{h-e}=n$. Denote by
$A=\bigcup\limits_{i=1}^t T_i=\{a_1,\ldots,a_n\}$. For any $a_i\in T_{j_0}$,  by Lemma \ref{Psi}, we have
\begin{equation}
\begin{aligned}
u_i=\Delta_{A}(a_i)=\Delta_{T_{j_0}}(a_i)\cdot\left(\prod_{j\neq j_0,j=1}^t\Psi_{T_j}(a_i)\right)=\prod_{j\neq j_0,j=1}^t(Tr(a_i)-b_j).
\end{aligned}
\end{equation}
From the definition of $Tr$, we know $Tr(a_i)\in\F_{p^e}$($1\leq i\leq n$). Since $b_j\in B\subseteq\F_{p^e}$($1\leq j\leq t$),
one has $u_i\in\F_{p^e}$. By the definition of $u_i$, it is easy to see that $u_i\neq 0$ with $1\leq i\leq n$. Therefore, for any $1\leq i\leq n$,
it takes $u_i\in\F_{p^e}^*\subseteq E$ by Lemma \ref{yC}. Let $v_{i}^{p^e+1}= u_i$($1\leq i\leq n$) and $s:=k-1-l$. Choose
\begin{center}
$\mathbf{a}=(a_1,\ldots,a_{n})$ and
$\mathbf{v}=(\alpha v_{1},\alpha v_{2},\ldots,\alpha v_{s},v_{s+1},\ldots,v_{n})$,
\end{center}
where $\alpha\in\mathbb{F}_{q}^*$ and $\alpha^{p^e+1}\neq 1$. Consider the $e$-Galois hull of the $[n+1,k]_q$ MDS code
$\mathcal{C}:=\mathbf{GRS}_{k}(\mathbf{a},\mathbf{v},\infty)$. Similarly as the proof of Theorem \ref{thmE}, we can obtain the result.
\end{proof}

\section{Constructions via MDS codes with special Euclidean orthogonal property}

\quad\; When $\frac{h}{e}$ is odd, we can make the constructions from MDS codes satisfying Euclidean orthogonal properties.

\begin{theorem}\label{via GRS}
Let $q=p^h$, where $p$ is an odd prime. Assume $1\leq m\leq \lfloor\frac{n}{2}\rfloor$ and $\frac{h}{e}$ is odd. Suppose
$$\mathbf{GRS}_{m}(\mathbf{a},\mathbf{v})^\perp=\mathbf{GRS}_{n-m}(\mathbf{a},\mathbf{v}).$$
For any $0\leq l\leq k \leq \lfloor\frac{p^e+n-1}{p^e+1}\rfloor$,
there exists a $q$-ary $[n,k]$ MDS code $\mathcal{C}$ with $\dim(Hull_e(\mathcal{C}))=l$.
\end{theorem}
\begin{proof}
From $\mathbf{GRS}_{m}(\mathbf{a},\mathbf{v})^\perp=\mathbf{GRS}_{n-m}(\mathbf{a},\mathbf{v})$ and Lemma \ref{cor1} (i),
\begin{equation}\label{Euclidean}
v_i^{2}=\lambda u_i\neq 0(1\leq i\leq n),
\end{equation}
where $\lambda\in \mathbb{F}_{q}^*$.

Since $\frac{h}{e}$ and $p$ are odd, then $\gcd(p^e+1,p^h-1)=2$. So there exist two integers $\mu$ and $\nu$, such that $\mu(p^e+1)+\nu(p^h-1)=2$.
Substituting into (\ref{Euclidean}), we have $v_i^{\mu(p^e+1)}=\lambda u_i\neq 0(1\leq i\leq n)$. Set $v'_i=v_i^\mu(1\leq i\leq n)$. Then
$${v'}_i^{p^e+1}=\lambda u_i\neq 0(1\leq i\leq n).$$

Denote by $s:=k-l$, $\mathbf{a}=(a_{1},a_{2},\ldots,a_{n})$ and
$\mathbf{v}^{'}=(\alpha v'_{1},\alpha v'_{2},\ldots,\alpha v'_{s},v'_{s+1},\ldots,v'_{n})$, where $\alpha\in\mathbb{F}_{q}^*$ and $\alpha^{p^e+1}\neq 1$.
Consider the $e$-Galois hull of the $[n,k]_q$ MDS code $\mathcal{C}:=\mathbf{GRS}_{k}(\mathbf{a},\mathbf{v}')$. Then for any
\begin{center}
$\mathbf{c}=(\alpha v'_{1}f(a_{1}),\ldots,\alpha v'_{s}f(a_{s}),v'_{s+1}f(a_{s+1}),\ldots,v'_{n}f(a_{n}))\in Hull_e(\mathcal{C})$
\end{center}
with $\deg(f(x))\leq k-1$, according to Lemma \ref{HE1} (i), there exists a polynomial $g(x)\in \mathbb{F}_{q}[x]$ with $\deg(g(x))\leq n-k-1$ such that
\begin{equation*}
\begin{aligned}
&(\alpha^{p^e+1} {v'}_{1}^{p^e+1}f^{p^e}(a_{1}),\ldots,\alpha^{p^e+1}{v'}_{s}^{p^e+1}f^{p^e}(a_{s}),{v'}_{s+1}^{p^e+1}f^{p^e}(a_{s+1}),
\ldots,{v'}_{n}^{p^e+1}f^{p^e}(a_{n}))&\\=&(u_{1}g(a_{1}),\ldots,u_{s}g(a_{s}),u_{s+1}g(a_{s+1}),\ldots,u_{n}g(a_{n})).&
\end{aligned}
\end{equation*}
Set $\xi=\alpha^{p^e+1}$. Then
\begin{equation}\label{equ11}
\begin{aligned}
&(\xi\lambda u_{1}f^{p^e}(a_{1}),\ldots,\xi\lambda u_{s}f^{p^e}(a_{s}),\lambda u_{s+1}f^{p^e}(a_{s+1}),\ldots,
\lambda u_{n}f^{p^e}(a_{n}))&\\=&(u_{1}g(a_{1}),\ldots,u_{s}g(a_{s}),u_{s+1}g(a_{s+1}),\ldots,u_{n}g(a_{n})).&
\end{aligned}
\end{equation}
Considering the last $n-s$ coordinates of (\ref{equ11}), we get $\lambda f^{p^e}(a_i)=g(a_i)$($s+1\leq i\leq n$). Hence the number of distinct roots
of $\lambda f^{p^e}(x)-g(x)$ is at least $n-s\geq n-k$. Since $k\leq \lfloor\frac{p^e+n-1}{p^e+1}\rfloor$, we have
$\deg(f^{p^e}(x))\leq p^e(k-1)\leq n-k-1$, which derives that $\deg(\lambda f^{p^e}(x)-g(x))\leq n-k-1$ together with
$\deg(g(x))\leq n-k-1$. Hence $\lambda f^{p^e}(x)=g(x)$.

Comparing the first $s$ coordinates of (\ref{equ11}), $$\xi\lambda u_if^{p^e}(a_i)= u_ig(a_i)=\lambda u_if^{p^e}(a_i)$$ for $i=1,\ldots,s$.
Hence $f(a_i)=0$ with $\xi\neq 1$ and $\lambda u_i\neq 0$ $(i=1,\ldots,s)$. Then $f(x)$ can be expressed as
\begin{equation*}
f(x)=c(x)\prod_{i=1}^{s}(x-a_{i}),
\end{equation*}
for some $c(x)\in \mathbb{F}_{q}[x]$ with $\deg(c(x))\leq k-1-s$. Therefore, $\dim(Hull_e(\mathcal{C}))\leq k-s$.

Conversely, put $f(x)=c(x)\prod\limits_{i=1}^{s}(x-a_{i})$, where $c(x)\in \mathbb{F}_{q}[x]$ and $\deg(c(x))\leq k-1-s$.
Assume that $g(x)=\lambda f^{p^e}(x)$, which yields $\deg(g(x))\leq n-k-1$. Then
\begin{equation*}
\begin{aligned}
&(\xi\lambda u_{1}f^{p^e}(a_{1}),\ldots,\xi\lambda u_{s}f^{p^e}(a_{s}),
\lambda u_{s+1}f^{p^e}(a_{s+1}),\ldots,\lambda u_{n}f^{p^e}(a_{n}))&\\
=&(u_{1}g(a_{1}),\ldots,u_{s}g(a_{s}),u_{s+1}g(a_{s+1}),\ldots,u_{n}g(a_{n})).&
\end{aligned}
\end{equation*}
By Lemma \ref{HE1} (i),
\begin{center}
$(\alpha v'_{1}f(a_{1}),\ldots,\alpha v'_{s}f(a_{s}), v'_{s+1}f(a_{s+1}),\ldots,  v'_{n}f(a_{n}))\in Hull_e(\mathcal{C})$.
\end{center}
Therefore, $\dim(Hull_e(\mathcal{C}))\geq k-s$.

Hence $\dim(Hull_e(\mathcal{C}))= k-s=l$.
\end{proof}

\begin{example}\label{eg1}
In Theorem 1 of [\ref{FLL}], we have $v_i^2=\lambda u_i$, where $\lambda=g^{\frac{\sqrt{q}+1}{2}}$ and $g$ is a primitive element of $\F_q$. By Lemma \ref{cor1},
$\mathbf{GRS}_{m}(\mathbf{a},\mathbf{v})^\perp=\mathbf{GRS}_{n-m}(\mathbf{a},\mathbf{v})$ is satisfied. Therefore, for any $0\leq l\leq k \leq \lfloor\frac{p^e+n-1}{p^e+1}\rfloor$,
there exists a $q$-ary $[n,k]$ MDS code $\mathcal{C}$ with $\dim(Hull_e(\mathcal{C}))=l$ by Theorem \ref{via GRS}.
\end{example}

Afterwards, we apply extended GRS codes to construct MDS codes with $e$-Galois hulls of arbitrary dimensions.

\begin{theorem}\label{via eGRS}
Let $q=p^h$($p$ is an odd prime), $\frac{h}{e}$ is odd and $1\leq m\leq \lfloor\frac{n+1}{2}\rfloor$. Suppose
$$\mathbf{GRS}_{m}(\mathbf{a},\mathbf{v},\infty)^{\bot}=\mathbf{GRS}_{n+1-m}(\mathbf{a},\mathbf{v},\infty).$$
For any $1\leq k \leq \lfloor\frac{p^e+n-1}{p^e+1}\rfloor$ and $0\leq l\leq k-1$, there exists a $q$-ary $[n+1,k]$ MDS code $\mathcal{C}$ with $\dim(Hull_e(\mathcal{C}))=l$.
\end{theorem}

\begin{proof}
(i). Since $\mathbf{GRS}_{m}(\mathbf{a},\mathbf{v},\infty)^{\bot}=\mathbf{GRS}_{n+1-m}(\mathbf{a},\mathbf{v},\infty)$ and by Lemma \ref{cor1} (ii), we obtain
\begin{equation}\label{eEuclidean}
v_{i}^{2}=-u_{i}\neq 0(1\leq i \leq n).
\end{equation}
It follows that $\gcd(p^e+1,p^h-1)=2$, since $\frac{h}{e}$ and $p$ are odd. Similar to Theorem \ref{via GRS}, it takes
$\mu(p^e+1)+\nu(p^l-1)=2$($\mu$ and $\nu$ are two integers) and $v_i^{\mu(p^e+1)}=-u_i\neq 0(1\leq i\leq n)$. Set $v'_i=v_i^\mu(1\leq i\leq n)$. Then
$${v'}_i^{p^e+1}= -u_i\neq 0(1\leq i\leq n).$$

Denote by $s:=k-l-1$, $\mathbf{a}=(a_{1},a_{2},\ldots,a_{n})$ and
$\mathbf{v}^{'}=(\alpha v'_{1},\alpha v'_{2},\ldots,\alpha v'_{s},v'_{s+1},\ldots,v'_{n})$, where $\alpha\in\mathbb{F}_{q}^*$ and $\alpha^{p^e+1}\neq 1$.
Consider the $e$-Galois hull of the $[n,k]_q$ MDS code $\mathcal{C}:=\mathbf{GRS}_{k}(\mathbf{a},\mathbf{v}',\infty)$. Then for any
\begin{center}
$\mathbf{c}=(\alpha v'_{1}f(a_{1}),\ldots,\alpha v'_{s}f(a_{s}),v'_{s+1}f(a_{s+1}),\ldots,v'_{n}f(a_{n}),f_{k-1})\in Hull_e(\mathcal{C})$
\end{center}
with $\deg(f(x))\leq k-1$, according to Lemma \ref{HE1} (ii), there exists a polynomial $g(x)\in \mathbb{F}_{q}[x]$ with $\deg(g(x))\leq n-k$ such that
\begin{equation*}
\begin{aligned}
&(\alpha^{p^e+1} {v'}_{1}^{p^e+1}f^{p^e}(a_{1}),\ldots,\alpha^{p^e+1}{v'}_{s}^{p^e+1}f^{p^e}(a_{s}),{v'}_{s+1}^{p^e+1}f^{p^e}(a_{s+1}),
\ldots,{v'}_{n}^{p^e+1}f^{p^e}(a_{n}),f_{k-1}^{p^e})&\\
=&(u_{1}g(a_{1}),\ldots,u_{s}g(a_{s}),u_{s+1}g(a_{s+1}),\ldots,u_{n}g(a_{n}),-g_{n-k}).&
\end{aligned}
\end{equation*}
Set $\xi=\alpha^{p^e+1}$. It yields
\begin{equation}\label{equ1}
\begin{aligned}
&(-\xi u_{1}f^{p^e}(a_{1}),\ldots,-\xi u_{s}f^{p^e}(a_{s}),-u_{s+1}f^{p^e}(a_{s+1}),\ldots,-u_{n}f^{p^e}(a_{n}),f_{k-1}^{p^e})&\\
=&(u_{1}g(a_{1}),\ldots,u_{s}g(a_{s}),u_{s+1}g(a_{s+1}),\ldots,u_{n}g(a_{n}),-g_{n-k}).&
\end{aligned}
\end{equation}
When $s+1 \leq  i \leq  n$, we get $-f^{p^e}(a_i)=g(a_i)$. Therefore, the number of distinct roots of $f^{p^e}(x)+g(x)$ is at least $n-s\geq n-k+1$.
We know $\deg(f^{p^e}(x))\leq p^e(k-1)\leq n-k-1$ from $k\leq \lfloor\frac{p^e+n-1}{p^e+1}\rfloor$, and $\deg(g(x))\leq n-k$.
Thus it derives that $\deg(f^{p^e}(x)+g(x))\leq n-k$. So $-f^{p^e}(x)=g(x)$.

From the first $s$ coordinates of (\ref{equ1}), we have $$-\xi u_if^{p^e}(a_i)=u_ig(a_i)= -u_if^{p^e}(a_i)$$ with $i=1,\ldots,s$, which derives
$f(a_i)=0$ with $\xi\neq 1$ and $u_i\neq 0$ $(i=1,\ldots,s)$. Then $f(x)$ can be expressed as
\begin{equation*}
f(x)=c(x)\prod_{i=1}^{s}(x-a_{i}),
\end{equation*}
for some $c(x)\in \mathbb{F}_{q}[x]$ with $\deg(c(x))\leq k-2-s$. It follows that $\dim(Hull_e(\mathcal{C}))\leq k-1-s$.

Conversely, put $f(x)=c(x)\prod\limits_{i=1}^{s}(x-a_{i})$, where $c(x)\in \mathbb{F}_{q}[x]$ and $\deg(c(x))\leq k-2-s$, which yields $f_{k-1}=0$.
Assume that $g(x)=-f^{p^e}(x)$. Then $\deg(g(x))\leq n-k-1$, which yields $g_{n-k}=0$. It takes
\begin{equation*}
\begin{aligned}
&(-\xi u_{1}f^{p^e}(a_{1}),\ldots,-\xi u_{s}f^{p^e}(a_{s}),-u_{s+1}f^{p^e}(a_{s+1}),\ldots,-u_{n}f^{p^e}(a_{n}),0)&\\
=&(u_{1}g(a_{1}),\ldots,u_{s}g(a_{s}),u_{s+1}g(a_{s+1}),\ldots,u_{n}g(a_{n}),0).&
\end{aligned}
\end{equation*}
According to Lemma \ref{HE1} (ii),
\begin{center}
$(\alpha v'_{1}f(a_{1}),\ldots,\alpha v'_{s}f(a_{s}), v'_{s+1}f(a_{s+1}),\ldots,  v'_{n}f(a_{n}),0)\in Hull_e(\mathcal{C})$.
\end{center}
Therefore, $\dim(Hull_e(\mathcal{C}))\geq k-1-s$.

As a result, $\dim(Hull_e(\mathcal{C}))= k-1-s=l$.
\end{proof}

\begin{example}\label{eg2}
In Theorem 1(ii) of [\ref{FZXF}], we have $v_i^2=-u_i$. By Lemma \ref{cor1},
$\mathbf{GRS}_{m}(\mathbf{a},\mathbf{v},\infty)^{\perp}=\mathbf{GRS}_{n+1-m}(\mathbf{a},\mathbf{v},\infty)$ is satisfied. Therefore,
for any $1\leq k \leq \lfloor\frac{p^e+n-1}{p^e+1}\rfloor$ and $0\leq l\leq k-1$, there exists a $q$-ary $[n+1,k]$ MDS code $\mathcal{C}$ with $\dim(Hull_e(\mathcal{C}))=l$.
\end{example}

\begin{remark}\label{remark}
These two theorems build relationships between $e$-Galois orthogonal property of GRS codes and Euclidean orthogonal property of GRS codes. Therefore, when $\frac{h}{e}$ is odd,
we can construct GRS codes with $e$-Galois hulls of arbitrary dimensions via GRS codes satisfying Euclidean orthogonal property. In fact, if an (extended) GRS code
is a Euclidean self-orthogonal code (including Euclidean self-dual code), the conditions of Theorems \ref{via GRS} and \ref{via eGRS} can be satisfied.
\end{remark}


\section{Applications to EAQECCs and MDS EAQECCs}

\quad\; In this section, we apply the results in Sections 3 and 4 to construct several families of EAQECCs and MDS EAQECCs,
which are more general than previous works. More details on EAQECCs are referred to [\ref{WB}].

An $[[n,k,d;c]]_q$ EAQECC $\mathcal{C}$ means that under the assist of $c$ pairs of maximally entangled Bell states, the quantum code $\mathcal{C}$ can
encode $k$ information qubits into $n$ channel qubits and $d$ represents the minimum distance. Similar to classical linear codes, EAQECC also satisfies
the quantum Singleton bound, which is given in the following lemma.

\begin{lemma}([\ref{BDH}, \ref{LA}, \ref{MFA}])\label{y6}
Assume that $d\leq\frac{n+2}{2}$. Then $[[n,k,d;c]]_{q}$ EAQECC satisfies
\begin{center}
$n+c-k\geq 2(d-1)$,
\end{center}
where $0\leq c\leq n-1$. \qquad\qquad\qquad\qquad\qquad\qquad\qquad\qquad\qquad\qquad\qquad\qquad\qquad\qquad\qquad\qquad\qquad $\square$
\end{lemma}

\begin{remark}
When $d\leq\frac{n+2}{2}$, an EAQECC attaining the quantum Singleton bound, that is $n+c-k=2(d-1)$, is called an MDS EAQECC.
\end{remark}

For a matrix $M=(m_{ij})$ over $\F_q$, define $M^{(p^{h-e})}=(m_{ij}^{p^{h-e}})$ and put $M^{T_e}=\left(M^{(p^{h-e})}\right)^T$.

In [\ref{HL}] and [\ref{QCWL}], the authors proposed methods for constructing EAQECCs by utilizing classical linear codes with $e$-Galois inner products
over finite fields in the following.

\begin{lemma}([\ref{HL}, Corollary 3.2] and [\ref{QCWL}, Corollary 5.2])\label{y7}
Let $H$ be a parity check matrix of an $[n,k,d]$ linear code $\mathcal{C}$ over $\F_{q}$. Then there exists an $[[n,2k-n+c,d;c]]_q$ EAQECC, where
$c=rank(HH^{T_e})$ is the required number of maximally entangled states.
\qquad\qquad\qquad\qquad\qquad\qquad\qquad\qquad\qquad\qquad\qquad\qquad\qquad\qquad\qquad\qquad\qquad\qquad $\square$
\end{lemma}

In [\ref{GJG}], Guenda et al. showed the relationship between the value of $rank(HH^{T_e})$ and the hull dimension of linear code with parity
check matrix $H$, that is $$rank(HH^{T_e})=n-k-\dim(Hull_e(\mathcal{C})).$$

As a direct consequence of Lemmas \ref{y6} and \ref{y7}, one has the following result, which has been shown in [\ref{QCWL}].

\begin{lemma}([\ref{QCWL}])\label{y8}
Let $\mathcal{C}$ be an $[n,k,d]$ linear code over $\F_q$ and its $e$-Galois dual $\mathcal{C}^{\perp_e}$ has parameters $[n,n-k,d^{\perp_e}]$. Then there exists
$[[n,k-\dim(Hull_e(\mathcal{C})),d;n-k-\dim(Hull_e(\mathcal{C}))]]_q$ EAQECC and
$[[n,n-k-\dim(Hull_{m-e}(\mathcal{C})),d^\perp;k-\dim(Hull_{m-e}(\mathcal{C}))]]_q$ EAQECC.
\end{lemma}

Let $\mathcal{C}$ be an $[n,k,n-k+1]$-MDS code over $\F_q$. Its $e$-Galois dual code $\mathcal{C}^{\bot_e}$ is also an MDS code with $[n,n-k,k+1]$.
Denote by $l=\dim(Hull_e(\mathcal{C}))$ and $l'=\dim(Hull_{m-e}(\mathcal{C}))$. Then we can obtain the following result by Lemma \ref{y8}.

\begin{corollary}([\ref{QCWL}])\label{cor8}
Assume that $\mathcal{C}$ is an $[n,k]$-MDS code over $\F_q$. If $k\leq \lfloor\frac{n}{2}\rfloor$, then there exists an $[[n,k-l,n-k+1;n-k-l]]_q$ EAQECC
and an $[[n,n-k-l,k+1;k-l]]_q$ MDS EAQECC. \qquad $\square$
\end{corollary}

From Corollary \ref{cor8} and all the theorems in Sections 3 and 4, we have the following results directly.

\begin{theorem}\label{EAQ2}
Let $q=p^{em}$ with $p$ odd prime. Assume that $t\mid(p^{e}-1)$, $m$ is even, $r\leq m-1$ and $n=tp^{er}$.

(i). For any $1\leq k\leq \lfloor\frac{p^e+n-1}{p^e+1}\rfloor$ and $0\leq l\leq k-1$, then there exists an $\big[[n+1,k-l,n-k+2;n+1-k-l]\big]_{q}$ EAQECC over $\F_q$.

(ii). For any $1\leq k\leq \lfloor\frac{p^{h-e}+n-1}{p^{h-e}+1}\rfloor$ and $0\leq l'\leq k-1$, then there exists an $\big[[n+1,n+1-k-l',k+1;k-l']\big]_{q}$ MDS EAQECC over $\F_q$.
\end{theorem}

\begin{example}
Choose $(p,m,r,e,t)=(5,4,3,3,31)$. It is easy to see that $t=31\mid 5^3-1=p^e-1$. By Theorem \ref{EAQ2}, there exist EAQECCs with parameters
$\big[[60546876,k-l,60546877-k;60546876-k-l]\big]_{q}$, where $1\leq k\leq 480531$ and $0\leq l\leq k-1$ and MDS EAQECCs with parameters
$\big[[60546876,60546876-k-l',k+1;k-l']\big]_{q}$, where $1\leq k\leq 31$ and $0\leq l'\leq k-1$. The two classes of EAQECCs are new in the sense that their parameters can not be covered by previous results.
\end{example}

\begin{theorem}\label{EAQ3}
Let $q=p^h$, $2e\mid h$ and $n=tp^{h-e}$ with $1\leq t\leq p^{e}$.

(i). For any $1\leq k\leq \lfloor\frac{p^e+n-1}{p^e+1}\rfloor$ and $0\leq l\leq k-1$,
there exists an $\big[[n+1,k-l,n-k+2;n+1-k-l]\big]_{q}$ EAQECC over $\F_q$.

(ii). For any $1\leq k\leq \lfloor\frac{p^{h-e}+n-1}{p^{h-e}+1}\rfloor$ and $0\leq l'\leq k-1$,
there exists an $\big[[n+1,n+1-k-l',k+1;k-l']\big]_{q}$ MDS EAQECC over $\F_q$.
\end{theorem}

\begin{example}
Choose $(p,h,e,t)=(3,16,4,73)$. Then $p^e=3^4=81>73=t$. By Theorem \ref{EAQ3}, there exist EAQECCs with parameters $\big[[38795194,k-l,38795195-k;38795194-k-l]\big]_{q}$,
where $1\leq k\leq 473113$ and $0\leq l\leq k-1$ and MDS EAQECCs with parameters $\big[[38795194,38795194-k-l',k+1;k-l']\big]_{q}$, where $1\leq k\leq 73$
and $0\leq l'\leq k-1$. These two classes of EAQECCs also have new parameters which have not been reported previously.
\end{example}

\begin{theorem}\label{EAQ4}
Let $q=p^h$. Assume $1\leq m\leq \lfloor\frac{n}{2}\rfloor$ and $\frac{h}{e}$ is odd. Suppose
$$\mathbf{GRS}_{m}(\mathbf{a},\mathbf{v})^\perp=\mathbf{GRS}_{n-m}(\mathbf{a},\mathbf{v}).$$

(i). For any $1\leq k \leq \lfloor\frac{p^e+n-1}{p^e+1}\rfloor$ and $0\leq l\leq k$, there exists an $\big[[n,k-l,n-k+1;n-k-l]\big]_{q}$ EAQECC over $\F_q$.

(ii). For any $1\leq k \leq \lfloor\frac{p^{h-e}+n-1}{p^{h-e}+1}\rfloor$ and $0\leq l'\leq k$, there exists an $\big[[n,n-k-l',k+1;k-l']\big]_{q}$ MDS EAQECC over $\F_q$.
\end{theorem}

\begin{example}
Similarly as Example \ref{eg1}, from Theorem 1 of [\ref{FLL}], we know there exists a $q$-ary $[n,\frac{n}{2}]$ MDS Euclidean self-dual code, where $n=s(\sqrt{q}-1)+t(\sqrt{q}+1)$, 
$\sqrt{q}$ and $s$ satisfy ``$\sqrt{q}\equiv 1 \pmod 4$ and $s$ is even" or ``$\sqrt{q}\equiv 3 \pmod 4$ and $s$ is odd". By Remark \ref{remark} and Theorem \ref{EAQ4}, 
there exists an $\big[[n,n-k-l',k+1;k-l']\big]_{q}$ MDS EAQECC or $\big[[n,k-l,n-k+1;n-k-l]\big]_{q}$ EAQECC over $\F_q$.
The length of the MDS EAQECC or EAQECC is $n=s(\sqrt{q}-1)+t(\sqrt{q}+1)$, which is more flexible than the previous results.
\end{example}

\begin{theorem}\label{EAQ5}
Let $q=p^h$, $\frac{h}{e}$ is odd and $1\leq m\leq \lfloor\frac{n+1}{2}\rfloor$. Suppose
$$\mathbf{GRS}_{m}(\mathbf{a},\mathbf{v},\infty)^{\perp}=\mathbf{GRS}_{n+1-m}(\mathbf{a},\mathbf{v},\infty).$$

(i). For any $1\leq k \leq \lfloor\frac{p^e+n-1}{p^e+1}\rfloor$ and $0\leq l\leq k-1$, there exists an $\big[[n+1,k-l,n-k+2;n+1-k-l]\big]_{q}$ EAQECC over $\F_q$.

(ii). For any $1\leq k \leq \lfloor\frac{p^{h-e}+n-1}{p^{h-e}+1}\rfloor$ and $0\leq l'\leq k-1$, there exists an $\big[[n+1,n+1-k-l',k+1;k-l']\big]_{q}$ MDS EAQECC over $\F_q$.
\end{theorem}

\begin{example}
Similarly as Example \ref{eg2}, from Theorem 1 of [\ref{FZXF}], there exists a $q$-ary MDS Euclidean self-dual code of length $n+1$. 
The parameters $n,n',n_1,n_2,t$ satisfy the following conditions:

(i). $n=tn'$ is odd;

(ii). $n'=n_1n_2$ and $n'\mid (q-1)$;

(iii). $n_1=\gcd(n',\sqrt{q}+1)$;

(iv). $n_2=\frac{n'}{\gcd(n',\sqrt{q}+1)}$;

(v). $1\leq t\leq \frac{\sqrt{q}-1}{n_2}$.

So by Remark \ref{remark} and Theorem \ref{EAQ5}, there exists an $\big[[n+1,n+1-k-l',k+1;k-l']\big]_{q}$ MDS EAQECC or an $\big[[n+1,k-l,n-k+2;n+1-k-l]\big]_{q}$ EAQECC
over $\F_q$. The length of the MDS EAQECC or EAQECC is $n+1=tn'+1$, which is more flexible than the previous results.
\end{example}

\begin{remark}
Theorems \ref{EAQ4} and \ref{EAQ5} propose a mechanism for the constructions of EAQECCs and MDS EAQECCs via MDS codes with Euclidean orthogonal property.
In fact, all Euclidean self-orthogonal (extended) GRS codes can be used to construct EAQECCs and MDS EAQECCs.
\end{remark}

\section{Conclusion}

\quad\; Inspired by [\ref{Cao}, \ref{QCWL}], we propose some new MDS codes with Galois hulls of arbitrary dimensions and several new families of EAQECCs
and MDS EAQECCs. In our constructions, the lengths of codes($n$ or $n+1$) are flexible. However, the dimension $k$ is roughly upper bounded by
$\lfloor\frac{p^e+n-1}{p^e+1}\rfloor$ or $\lfloor\frac{p^{h-e}+n-1}{p^{h-e}+1}\rfloor$. How to increase the upper bound is not an easy task, which is left as an open problem.

When $\frac{h}{e}$ is odd, we can associate Galois hulls of GRS codes with GRS codes satisfying Euclidean orthogonal property. Precisely, if there exists
an (extended) GRS code satisfying one of the Euclidean orthogonal properties of Theorems \ref{via GRS} and \ref{via eGRS}, then we can construct an
$[n,k]_q$ (extended) GRS code $\mathcal{C}$ with $\dim(Hull_e(\mathcal{C}))=l$, where $0\leq l\leq k-1$. For the case $\frac{h}{e}$ is even, some
constructions are given. However, how to propose a mechanism for the constructions of MDS codes with Galois hulls of arbitrary dimensions is also an open problem.

\section*{Acknowledgements}
{The authors thank anonymous reviewers, the editor and the associate editor for their suggestions and comments to improve the readability of this paper.
This research is supported by National Natural Science Foundation of China under Grant 11471008 and the Fundamental Research Funds for the Central Universities of CCNU under grant CCNU20TD002.
}

\end{document}